\documentclass{llncs}
\pdfoutput=1

\def\SPRINGER{}

\usepackage{mathtools}
\usepackage{thmtools}
\usepackage{amssymb}
\usepackage{stmaryrd}

\ifdefined\SPRINGER
\else
  \usepackage{amsthm}
\fi

\usepackage{newtxtext}
\usepackage{newtxmath}
\usepackage{inconsolata}

\let\phi\varphi
\let\epsilon\varepsilon
\let\emptyset\varnothing

\newcommand{\limpl}{\Rightarrow}

\newcommand{\bigand}{\bigwedge}
\newcommand{\bigor}{\bigvee}
\newcommand{\union}{\cup}
\newcommand{\bigunion}{\bigcup}
\newcommand{\isect}{\cap}

\newcommand{\join}{\sqcup}

\newcommand{\meet}{\sqcap}

\newcommand{\pset}[1]{\mathcal{P}(#1)}

\newcommand{\true}{\mathit{true}}

\newcommand{\lfun}[2]{\lambda#1.#2}

\newcommand{\lfpName}{\operatorname{lfp}}
\newcommand{\lfpoName}[1]{\lfpName_{#1}}
\newcommand{\lfpo}[2]{\lfpoName{#1}\!#2}
\newcommand{\lfplo}[3]{\lfpo{#1}{\lambda#2.#3}}

\newcommand{\assume}[2][]{\ifthenelse{\equal{#1}{}}{[#2]}{[#2,#1]}}

\newcommand{\abstrName}[1][]{\alpha^{#1}}
\newcommand{\abstr}[2][]{\abstrName[#1](#2)}
\newcommand{\concreteName}[1][]{\gamma_{#1}}
\newcommand{\concrete}[2][]{\concreteName[#1](#2)}

\newcommand{\postName}[1][]{\operatorname{post}^{#1}}
\newcommand{\post}[3][]{\postName[#1](#2,#3)}
\newcommand{\preName}[1][]{\operatorname{pre}^{#1}}
\newcommand{\pre}[3][]{\preName[#1](#2,#3)}

\newcommand{\D}[1][]{\mathbb{D}_{#1}}
\newcommand{\Dbot}[1][]{\bot_{#1}}
\newcommand{\Dtop}[1][]{\top_{#1}}
\newcommand{\Dleq}[1][]{\sqsubseteq_{#1}}

\newcommand{\Dgeq}[1][]{\sqsupseteq_{#1}}

\newcommand{\Djoin}[1][]{\join_{#1}}

\newcommand{\Dmeet}[1][]{\meet_{#1}}

\newcommand{\DpostName}[1][\sharp]{\postName[#1]}
\newcommand{\Dpost}[3][\sharp]{\post[#1]{#2}{#3}}
\newcommand{\DpreName}[1][\sharp]{\preName[#1]}

\ifdefined\SPRINGER
\else
  \theoremstyle{definition}

  \newtheorem{proposition}{Proposition}
\fi

\usepackage{ifthen}
\usepackage{paralist}
\usepackage{wrapfig}
\usepackage{tikz}
\usepackage{hyperref}
\hypersetup{final}

\title{Combining Forward and Backward Abstract Interpretation of Horn Clauses
\thanks{This work was partially supported by the
\href{http://erc.europa.eu/}{European Research Council} under the
European Union's Seventh Framework Programme (FP/2007-2013) / ERC Grant
Agreement nr. 306595 \href{http://stator.imag.fr}{``STATOR''}.}}
\author{%
Alexey Bakhirkin%
\and David Monniaux
}
\institute{Univ. Grenoble Alpes, VERIMAG, F-38000 Grenoble, France\\
CNRS, VERIMAG, F-38000 Grenoble, France}
\date{}

\renewcommand{\v}[1]{\mathbf{#1}}
\newcommand{\hcbody}[2]{\ifthenelse{\equal{#1}{} \or \equal{#2}{}}{#2#1}{#2,#1}}
\newcommand{\hc}[3]{#3 \leftarrow \hcbody{#1}{#2}}
\newcommand{\pfalse}{\mathfrak{f}}

\newcommand{\Atoms}{\mathbb{A}}
\newcommand{\HSys}{\mathbb{H}}
\newcommand{\TRel}[1][\HSys]{\mathbb{T}_{#1}}
\newcommand{\TInit}[1][\HSys]{\mathbb{I}_{#1}}
\newcommand{\TStep}[1][\HSys]{\mathbb{T}_{#1}^{\rightarrow}}
\newcommand{\TPred}{T_{\Pi}}
\newcommand{\preR}[3]{\preName\!|_{#1}(#2,#3)}
\newcommand{\DpreR}[3]{\DpreName\!|_{#1}(#2,#3)}
\newcommand{\goal}{A_g}
\newcommand{\dgoal}{g}
\newcommand{\Pgoal}{\Pi_g}
\newcommand{\tgoal}{T_g}
\newcommand{\df}{d}
\newcommand{\db}{b}

\newcommand{\leaf}[1]{\mathit{leaf}(#1)}
\newcommand{\tree}[2]{\mathit{tree}(#1 \leftarrow #2)}
\renewcommand{\root}[1]{\mathit{root}(#1)}
\newcommand{\isnode}[2]{\mathit{isnode}(#1,#2)}
\newcommand{\leaves}[1]{\mathit{leaves}(#1)}

\newcommand{\TpostName}{\postName[t]}
\newcommand{\Tpost}[2]{\post[t]{#1}{#2}}

\newcommand{\Tpre}[2]{\pre[t]{#1}{#2}}
\newcommand{\Tabstr}[1]{\abstr[t]{#1}}
\newcommand{\TabstrName}{\abstrName[t]}

\newcommand{\bparagraph}[1]{\par\vspace{.5\baselineskip}\noindent\textbf{#1} }

\newcounter{exmplNo}
\newtoks\exmplTitle
\newenvironment{exmpl}[2][]{%
  \refstepcounter{exmplNo}%
  \ifthenelse{\equal{#1}{}}{%
    \exmplTitle={}%
  }{%
    \exmplTitle={\ -- #1}%
  }%
  \bparagraph{Example \theexmplNo \the\exmplTitle.}%
  \label{#2}%
}{}

\usepackage[final]{listings}
\allowdisplaybreaks

\begin{document}

\maketitle

\ifdefined\HCVS
  \tikz[remember picture,overlay]
  \node[shape=rectangle,fill=red!10,minimum width=\paperwidth,anchor=north west] 
  at (current page.north west) {We submit this paper for the presentation-only 
  category. If was originally submitted and accepted to SAS'2017 conference.};
\else
\fi

\begin{abstract}
  Alternation of forward and backward analyses is a standard technique in 
  abstract interpretation of \emph{programs}, which is in particular useful when 
  we wish to prove unreachability of some undesired program states.
  The current state-of-the-art technique for combining forward (bottom-up, in 
  logic programming terms) and backward (top-down) abstract interpretation of 
  \emph{Horn clauses} is query-answer transformation.
  It transforms a system of Horn clauses, such that standard forward analysis 
  can propagate constraints both forward, and backward from a goal.
  Query-answer transformation is effective, but has issues that we wish to 
  address.
  For that, we introduce a new backward collecting semantics, which is suitable 
  for alternating forward and backward abstract interpretation of Horn clauses.
  We show how the alternation can be used to prove unreachability of the goal 
  and how every subsequent run of an analysis yields a refined model of the 
  system.
  Experimentally, we observe that combining forward and backward analyses is 
  important for analysing systems that encode questions about reachability in C 
  programs.
  In particular, the combination that follows our new semantics improves the 
  precision of our own abstract interpreter, including when compared to a 
  forward analysis of a query-answer-transformed system.
\end{abstract}

\section{Introduction}

In the past years, there has been much interest in using Horn clauses for 
program analysis, i.e., to encode the program semantics and the analysis 
questions as a system of Horn clauses and then use a dedicated Horn clause 
solver to find a model of the system or show its unsatisfiability (see e.g., 
\cite{bgmr15horn-verification}).
In particular, collecting semantics of programs and reachability questions can 
be encoded as \emph{constrained Horn clauses}, or \emph{CHCs}.

With this approach, Horn clauses become a common language that allows different 
tools to exchange program models, analysis questions and analysis results.
For example, as part of this work, we implemented a polyhedra-based abstract 
interpreter for CHCs.
We use an existing tool SeaHorn \cite{gurfinkel15seahorn} to convert questions 
about reachability in C programs into systems of CHCs, and
this way we can use our abstract interpreter to analyse numeric C programs 
without having to ourselves implement the semantics of C.
Additionally, Horn clauses allow to build complicated abstract models of 
programs, as opposed to implementing the abstraction mostly as part of the 
abstract domain.
For example, D. Monniaux and L. Gonnord propose \cite{david16cell-morphing} a 
way to abstract programs that use arrays into array-free Horn clauses, and we 
are not aware of a domain that implements their abstraction.

On the other hand, this approach makes it more important to implement different 
precision-related techniques and heuristics in the analyser, since we have 
little control over how the problem description is formulated, when it is 
produced by an external procedure.
One technique that is important for disproving reachability using abstract 
interpretation is the combination of forward and backward analyses.
The idea is to alternate forward and backward analyses, and build an 
over-approximation of the set of states that are both reachable from the program 
entry and can reach an undesired state (Patrick and Radhia Cousot give a good 
explanation of the technique \cite[section 4]{cousot99refining-mc}).

Patrick and Radhia Cousot also propose to use a combination of forward and 
backward analyses a for logic programs \cite{cousot92logic-programs}.
Their combination is based on the intersection of forward (bottom-up, in logic 
programming terms\footnote{%
In this paper, we use the terms \emph{bottom-up} and \emph{top-down} in the 
meanings that they bear in logic programming and thus they correspond to 
\emph{forward} and \emph{backward} analysis respectively.
In program analysis, \emph{bottom-up} may mean \emph{from callees to callers} or 
\emph{from children to parents in the AST}, but this is \emph{not} the meaning 
that we intend in this paper.
})
and backward (top-down) collecting semantics, which, as we observe in 
\autoref{sec:main}, is too over-approximate for our purposes.
The current state-of-the-art technique for combining forward and backward 
analyses of Horn clauses is query-answer transformation \cite{gallagher15qa}.
The idea is to transform a system of Horn clauses, such that standard forward 
analysis can propagate constraints both forward from the facts, and backward 
from a goal.
Query-answer transformation is effective, e.g., B. Kafle and J. P. Gallagher 
report \cite{gallagher15qa} that it increases the number of benchmark programs 
that can be proven safe both by their abstract interpreter and by a pre-existing 
CEGAR-based analyser.
Still, query-answer transformation has some issues, which we outline  (together 
with its advantages) in \autoref{sec:background/qa} and revisit in 
\autoref{sec:main/revisiting-qa}.

To address the issues of the existing techniques, we introduce a new backward 
collecting semantics of CHCs, which offers more precision when combining forward 
and backward abstract interpretation.
We show how the analysis based on the new semantics can be used to prove 
unreachability of a goal and how every subsequent run of the analysis yields a 
refined model of the system.
In particular, if the goal is proven to be unreachable, our analysis can produce 
a model of the system that is disjoint from the goal, which allows to check the 
results of the analysis and to communicate them to other tools.
These are the main contributions of this paper.
To evaluate our approach, we take programs from the categories ``loops'', and 
``recursive'' of the Competition on Software Verification SV-COMP 
\cite{url-sv-comp}.
We use the existing tool SeaHorn to translate these programs to systems of Horn 
clauses.
We observe that the alternation of forward and backward analyses following our 
new semantics improves the precision of our own
abstract interpreter (i.e., it allows to prove safety of more safe programs) 
including when compared to forward analysis of a query-answer-transformed 
system.

\section{Background}
\label{sec:background}

We say that a \emph{term} is a variable, a constant, or an application of an 
\emph{interpreted} function to a vector of terms.
To denote vectors of terms, we use bold letters.
Thus, $\v{t}$ denotes a vector of terms;
$\phi[\v{x}]$ (assuming elements of $\v x$ are distinct) denotes a formula 
$\phi$, where the set of free variables is the set of elements of $\v{x}$; and
$\phi[\v{x}/\v{t}]$ denotes a formula that is obtained from $\phi$ by 
simultaneously replacing (substituting) every occurrence of $x_i \in \v{x}$ with 
the corresponding element $t_i \in \v{t}$.

\bparagraph{CHCs.}
A constrained Horn clause (CHC) is a first order formula of the form
\[
\forall X.\big(\,p_1(\v{t_1}) \land p_2(\v{t_2}) \land \cdots \land p_n(\v{t_n}) 
\land \phi \limpl p_{n+1}(\v{t_{n+1}})\,\big)
\]
where $p_i$ are uninterpreted predicate symbols, $\v t_i$ are vectors of terms;
$\phi$ is a quantifier-free formula in some background theory and does not 
contain uninterpreted predicates or uninterpreted functions; and $X$ includes 
all free variables of the formula under the quantifier.
Following standard notation in the literature, we write a Horn clause as
\[
\hc{p_1(\v{t_1}),p_2(\v{t_2}),\cdots,p_n(\v{t_n})}{\phi}{p_{n+1}(\v{t_{n+1}})}
\]
that is, with free variables being \emph{implicitly} universally quantified.
We use a capital letter to denote an \emph{application} of a predicate to 
\emph{some} vector of terms (while for predicate symbols, we use lowercase 
letters).
Thus, when the terms in predicate applications are not important, we can write 
the above clause as
\[
\hc{P_1,P_2,\cdots,P_n}{\phi}{P_{n+1}}
\]

The predicate application $P_{n+1}$ is called the \emph{head} of the clause, and 
the conjunction $\hcbody{P_1,P_2,\cdots,P_n}{\phi}$ is called the \emph{body}.
A CHC always has a predicate application as its head.
\emph{But}, we assume that there exists a distinguished 0-ary predicate 
$\pfalse$ that denotes falsity and is only allowed to appear in the head of a 
clause.
A clause that has $\pfalse$ as its head is called an \emph{integrity 
constraint}.
For example, an assertion $\hc{P}{\phi}{\psi}$ can be written as the integrity 
constraint: $\hc{P}{(\phi \land \neg \psi)}{\pfalse}$.

A \emph{system} is a set of CHCs that is interpreted as their conjunction.

\bparagraph{Models of CHCs.}
We say that an \emph{atom} is a formula of the form $p(c_1,\cdots,c_n)$, where 
$p$ is an n-ary predicate symbol and $c_i$ are constants.
We denote the set of all atoms by $\Atoms$.

An \emph{interpretation} is a set of \emph{atoms} $M \subseteq \Atoms$.
One can say that an interpretation gives \emph{truth assignment} to atoms: an 
atom is interpreted as \emph{true} if it belongs to the interpretation and as 
\emph{false} otherwise.
This way, an interpretation also provides a truth assignment to every formula, 
by induction on the formula structure.

For a system of CHCs, a \emph{model} (or \emph{solution}) is an interpretation 
that that makes every clause in the system $\true$ (note that all variables in a 
system of Horn clauses are universally quantified, and thus the model does not 
include variable valuations).
We call a model $M \subseteq \Atoms$ \emph{safe} when $\pfalse \notin M$ (many 
authors prefer to call an interpretation $M$ a model only when it does not 
include $\pfalse$, but we prefer to have both notions).
A system of CHCs always has the \emph{minimal} model w.r.t. subset ordering 
(see, e.g., \cite[section 4]{jaffar94clp-survey}).
If a system has no clauses of the form $\hc{}{\phi}{P}$, its least model is 
$\emptyset$.
We call a system of CHCs safe iff it has a safe model.
In particular, for a safe system, its least model is safe, and thus, for a safe 
system, there exists the smallest safe model.
For every system of CHCs, the set of atoms $\Atoms$ is the greatest (unsafe) 
model, but a safe system in general may not have the greatest safe model.

\bparagraph{Fixed Point Characterization of the Least Model.}
A system of CHCs $\HSys$ induces the \emph{direct consequence relation} $\TRel 
\subseteq \pset{\Atoms}\times\Atoms$, which is constructed as follows.
A tuple $\big(\{p_1(\v{c_1}),\cdots,p_n(\v{c_n})\}, p_{n+1}(\v{c_{n+1}})\big) 
\in \TRel$ iff the system $\HSys$ contains a clause 
$\hc{p_1(\v{t_1}),\cdots,p_n(\v{t_n})}{\phi}{p_{n+1}(\v{t_{n+1}})}$, such that 
$\phi \land \bigand_{i=1}^{n+1} \v{c_i} = \v{t_i}$ is satisfiable.\footnote{%
There may be a slight abuse of notation here.
When writing down the set as $\{p_1(\v{c_1}),\cdots,p_n(\v{c_n})\}$, we do 
\emph{not} assume that all $p_i$ or all $\v{c_i}$ are distinct and that the set 
has exactly $n$ elements.
}
In particular, every clause of the form $\hc{}{\phi}{p(\v{t})}$ induces a set of 
\emph{initial transitions} (or \emph{initial consecutions}) of the form 
$(\emptyset, p(\v{c}))$, where $\phi \land (\v{c} = \v{t})$ is satisfiable.
Direct consequence relation can be seen as a variant of a direct consequence 
function discussed by J.  Jaffar and M. J.  Maher \cite[section 
4]{jaffar94clp-survey}.

Note that $\TRel$ is unlike an ordinary transition relation and relates a 
\emph{set} of atoms with a single atom that is their direct consequence.
To work with such a relation, we can adapt the standard in program analysis 
definition of post-condition as follows:
\[
\post{\TRel}{X} = \{ a' \mid \exists A \subseteq X.\,(A,a') \in \TRel \}
\]
Then, the least model of $\HSys$ can be characterised as the least fixed point:
\begin{equation}
  \label{eq:forward-semantics}
  \lfplo{\subseteq}{X}{\post{\TRel}{X}}
\end{equation}
As standard in abstract interpretation, we call the fixed point 
\eqref{eq:forward-semantics} the forward (bottom-up, in logic programming terms) 
collecting semantics of $\HSys$.
In general, every pre-fixpoint of the consequence operator, i.e., every set $M$, 
s.t. $\post{\TRel}{M} \subseteq M$ is a model of $\HSys$.

\bparagraph{Analysis Questions.}
Given a system of CHCs $\HSys$, the analysis question may be stated in a number 
of ways.
Often we want to know whether the system is safe, i.e., whether the least model 
of $\HSys$ contains $\pfalse$.
More generally, we may be given a set of goal atoms $\goal \subseteq \Atoms$.
Then, the analysis question will be whether the goal is unreachable, i.e.  
whether the goal and the least model are disjoint.
In this case, we start by computing a (reasonably small) model $M$ of $\HSys$.
If $M \isect \goal = \emptyset$, we conclude that the goal is unreachable.
Otherwise, we either report an inconclusive result (since the computed $M$ will 
in general not be the smallest model), or attempt to compute a refined model $M' 
\subset M$.

Alternatively, we may want to produce a model of $\HSys$ that gives us some 
non-trivial information about the object described by $\HSys$.
In this case, we usually want to produce some reasonably small model, which is 
what abstract interpretation tries to do.
The goal may or may not be given.
For example, we may be only interested in some part of the object (say, a subset 
of procedures in a program), which is described by a subset of predicates $\Pi$.
Then, the goal will be the corresponding set of atoms $\goal = \{ p(\v c) \mid p 
\in \Pi\}$.

\subsection{Abstract Interpretation of CHCs}
\label{sec:background/chc-ai}

Abstract interpretation \cite{cousot77abst-int} provides us a way to 
\emph{compute an over-approximation} of the least model, following the fixed 
point characterization.
To do so, we introduce the \emph{abstract domain} $\D$ with the least element 
$\Dbot$, greatest element $\Dtop$, partial order $\Dleq$ and join $\Djoin$.
Every element of the abstract domain $d \in \D$ represents the set of atoms 
$\concrete{d} \subseteq \Atoms$.
Then, we introduce the abstract consequence operator $\DpostName$ which 
over-approximates the concrete operator $\postName$, i.e., for every $d \in \D$,
$
\concrete{\Dpost{H}{d}} \supseteq \post{\TRel}{\concrete{d}}
$.
If we are able to find such element $d_m \in \D$ that $\Dpost{H}{d_m} \Dleq d_m$ 
then $\concrete{d_m}$ is a pre-fixpoint of the direct consequence operator and 
thus a model of $\HSys$ (not necessarily the smallest one).
At this point, it does not matter how we compute $d_m$.
It may be a limit of a Kleene-like iteration sequence (as in our implementation) 
or it may be produced by policy iteration 
\cite{george16policy-iteration,seidl12policy-iteration}, etc.

One can expect that an element $d \in \D$ is partitioned by predicate, in the 
same way as in program analysis, domain elements are partitioned by program 
location.
In the simple case, every element $d \in \D$ will have a logical representation 
in some theory and one can think that it maps every predicate $p_i$ to a 
quantifier-free formula $\delta_i[\v{x_i}]$, where $\v{x_i}$ correspond to the 
arguments of $p_i$.
For example, when using a polyhedral domain, $d$ will map every predicate to a 
conjunction of linear constraints.
For simplicity of syntactic manipulations, we can assume that $\v{x_i}$ are 
distinct vectors of distinct variables, i.e., a given variable appears only in 
one vector $\v{x_i}$ and only once.

From this, we can derive a recipe for Kleene-like iteration.
Let $d \in \D$ be the current fixpoint candidate that maps every predicate $p_i$ 
to a formula $\delta_i[\v{x_i}]$.
We try to find a clause 
$\hc{p_1(\v{t_1}),\cdots,p_n(\v{t_n})}{\phi}{p_{n+1}(\v{t_{n+1}})}$ (where $n 
\geq 0$), such that the following formula is satisfiable:
\begin{equation}
  \label{eq:hc-negation}
  \phi \land \delta_1[\v{x_1}/\v{t_1}] \land \cdots \land 
  \delta_n[\v{x_n}/\v{t_n}] \land \neg \delta_{n+1}[\v{x_{n+1}}/\v{t_{n+1}}]
\end{equation}
If it is, we find a \emph{set} of models of \eqref{eq:hc-negation}, and if some 
model assigns the vector of constants $\v{c_{n+1}}$ to the variables 
$\v{x_{n+1}}$, we join the atom $p_{n+1}(\v{c_{n+1}})$ to $d$.
In a polyhedral analysis, we usually want to find in every step a \emph{convex}  
set models of \eqref{eq:hc-negation}.
Assuming the formula is in negation normal form, there is a na\"ive way to 
generalize a single model to a convex set of models by recursively traversing 
the formula and collecting atomic propositions satisfied by the model
(descending into all sub-formulas for a conjunction and into one sub-formula for 
a disjunction).
In general though, this corresponds to a problem of finding a model of a Boolean 
formula that is in some sense optimal  (see, e.g., the work of J.  Marques-Silva 
et al. \cite{silva13minimal-sets}).
When the set of CHCs is produced from a program by means of large block encoding 
\cite{beyer09lbe} (e.g., SeaHorn does this by default), then $\phi$ is 
disjunctive and represents some set of paths through the original program.
Finding a convex set of models of \eqref{eq:hc-negation} corresponds to finding 
a path through the original program, along which we need to propagate the 
post-condition.
In program analysis, a similar technique is called \emph{path focusing}  
\cite{david11path-focusing,david12pagai}.

\bparagraph{Checking the Model.}
Given an element $d \in \D$, we can check whether it represents a model by 
taking its abstract consequence.
If $\Dpost{\HSys}{d} \Dleq d$ then $\concrete{d}$ is a pre-fixpoint of the 
direct consequence operator and thus is a model of $\HSys$.
When $d$ can be represented in a logical form and maps every predicate $p_i$ to 
a formula $\delta_i[\v{x_i}]$ in some theory, we can check whether it represents 
a model (i.e., that for every clause, the formula
\eqref{eq:hc-negation} is unsatisfiable) using an SMT solver.
Being able to check the obtained models provides a building block for making a 
\emph{verifiable} static analyser.

\subsection{Program Analysis and CHCs}
\label{sec:background/program-analysis-chcs}

Different flavours of Horn clauses can be used to encode in logic form different 
program analysis questions.
In particular, CHCs can be used to encode invariant generation and reachability 
problems.
In such an encoding, uninterpreted predicates typically denote sets of reachable 
memory states at different program locations, clauses of the form 
$\hc{P_1,P_2,\cdots,P_n}{\phi}{P_{n+1}}$ encode the semantics of transitions 
between the locations, clauses of the form $\hc{}{\phi}{P}$ encode the initial 
states, and the integrity constraints (of the form $\hc{P}{\phi}{\pfalse}$) 
encode the assertions.
In this paper, we limit ourselves to invariant generation and reachability, but 
other program analysis questions (including verification of temporal properties
\cite{rybalchenko13existentially-quantified}) can be encoded using other 
flavours of Horn clauses.
For more information, an interested reader can refer to a recent survey 
\cite{bgmr15horn-verification}.

\newsavebox{\parallelIncrementText}
\begin{lrbox}{\parallelIncrementText}
  \begin{lstlisting}
x = y = 0;
while (*) {
  if (x $\geq$ 0) {
    x += 1; y += 1;
  } else {
    x += 1;
  }
}
assert(x == y);
  \end{lstlisting}
\end{lrbox}

\begin{figure}[t]
  \centering
  \begin{minipage}[b]{.45\textwidth}
    \centering
    \usebox{\parallelIncrementText}
    \caption{A program that increments $x$ and $y$ in parallel.}
    \label{fig:parallel-increment-text}
  \end{minipage}
  \hfill
  \begin{minipage}[b]{.5\textwidth}
    \centering
    \[
    \begin{split}
      & \hc{}{x = 0 \land y = 0}{p(x, y)} \\
      & \hc{p(x, y)}{x \geq 0}{p(x+1, y+1)} \\
      & \hc{p(x, y)}{x < 0}{p(x+1, y)} \\
      & \hc{p(x, y)}{x \neq y}{\pfalse}
    \end{split}
    \]
    \caption{Horn clause encoding of the program in 
    \autoref{fig:parallel-increment-text}.  }
    \label{fig:parallel-increment-chc}
  \end{minipage}
\end{figure}

\begin{exmpl}[Parallel Increment]{ex:parallel-increment}
  Consider a program in \autoref{fig:parallel-increment-text}.
  It starts by setting two variables, $x$ and $y$, to zero and then increments 
  both of them in a loop a non-deterministic number of times.
  An analyser is supposed to prove that after the loop finishes, $x$ and $y$ 
  have equal values.
  This program also has an unreachable condition $x < 0$ upon which only $x$ is 
  incremented, which will be useful in the next example.
  The program  in \autoref{fig:parallel-increment-text} can be encoded into CHCs 
  as shown in \autoref{fig:parallel-increment-chc}, where the predicate $p$ 
  denotes the set of reachable states at the head of the loop, and its arguments 
  denote the variables $x$ and $y$ respectively.
  From the point of view of abstract interpretation, such a system of CHCs 
  represents a program's collecting semantics.
  For simple programs, as the one in \autoref{fig:parallel-increment-text}, a 
  model of the system of CHCs directly represents an inductive invariant of the 
  program.
  For the more complicated programs (e.g., programs with procedures) this may no 
  longer be true, but in any case, if we find a safe (not containing $\pfalse$) 
  model of the system of CHCs, we can usually conclude that the program cannot 
  reach an assertion violation.
  A model that we find with abstract interpretation will assign to every 
  predicate an element of some abstract domain;
  for a numeric program this may be a convex polyhedron (or a small number of 
  polyhedra) in a space where every dimension corresponds to a predicate 
  argument.
  Thus, for us to be able to prove safety of a program, the system of CHCs has 
  to have a safe model of the given form.
\end{exmpl}

Horn clause encoding of programs without procedures is typically straightforward 
and results in a system, where every clause has at most one predicate 
application in the body; such clauses are often called \emph{linear}.
Encoding of programs with procedures is also possible, but there are multiple 
ways of doing it.
We now give an example of a program with a procedure.

\newsavebox{\parallelIncrementInterprocText}
\begin{lrbox}{\parallelIncrementInterprocText}
  \begin{lstlisting}
void inc_xy() {
  if (x $\geq$ 0) {
    x += 1; y += 1;
  } else {
    x += 1;
  }
}
$\ldots$
x = y = 0;
while (*)
  inc_xy();
assert(x == y);
  \end{lstlisting}
\end{lrbox}

\begin{figure}[t]
  \centering
  \begin{minipage}[b]{.45\textwidth}
    \centering
    \usebox{\parallelIncrementInterprocText}
    \caption{A program that increments $x$ and $y$ in parallel using a 
    procedure.}
    \label{fig:parallel-increment-interproc-text}
  \end{minipage}
  \hfill
  \begin{minipage}[b]{.5\textwidth}
    \centering
    \[
    \begin{split}
      & \hc{}{x = 0 \land y = 0}{p(x, y)} \\
      & \hc{p(x, y)}{f(x, y, x', y')}{p(x', y')} \\
      & \hc{p(x, y)}{x \neq y}{\pfalse} \\
      & \hc{f_c(x, y)}{x \geq 0}{f(x, y, x+1, y+1)} \\
      & \hc{f_c(x, y)}{x < 0}{f(x, y, x+1, y)} \\
      & \hc{}{\true}{f_c(x, y)}
    \end{split}
    \]
    \caption{A possible Horn clause encoding of the program in 
    \autoref{fig:parallel-increment-interproc-text}.  }
    \label{fig:parallel-increment-interproc-chc}
  \end{minipage}
\end{figure}

\begin{exmpl}[Parallel Increment Using a 
  Procedure]{ex:parallel-increment-interproc}
  Consider a program in \autoref{fig:parallel-increment-interproc-text}.
  Similarly to Example~\ref{ex:parallel-increment}, it starts by setting two 
  variables, $x$ and $y$, to zero and then increments both of them in a loop, 
  but this time by calling an auxiliary procedure.
  Again the procedure has an unreachable condition $x < 0$ upon which it only 
  increments $x$.
  If we encode this program into CHCs directly (without inlining of 
  \texttt{inc\_xy}), we may arrive at a system as in 
  \autoref{fig:parallel-increment-interproc-chc}.
  This roughly corresponds to how the tool SeaHorn encodes procedures that do 
  not contain assertions.
  As before, the predicate $p$ denotes the reachable states at the loop head.
  A new predicate $f$ denotes the \emph{input-output relation} of the procedure 
  \texttt{inc\_xy}.
  If $f(x_1,y_1,x_2,y_2)$ holds, this means that if at the entry of  
  \texttt{inc\_xy} $x = x_1$ and $y = y_1$ then at the exit of \texttt{inc\_xy}, 
  it \emph{may} be the case that $x = x_2$ and $y = y_2$.
  In general, every predicate that corresponds to a location inside a procedure, 
  will have two sets of arguments: one set will correspond to the state at the 
  entry of the procedure (as the first two arguments of $f$) and the other, to 
  the corresponding state at the given location (as the last two arguments of 
  $f$).
  Note that another new predicate, $f_c$, is purely auxiliary and does 
  \emph{not} denote the reachable states at the at the initial location of 
  \texttt{inc\_xy}.
  To solve the system in \autoref{fig:parallel-increment-interproc-chc}, we need 
  to approximate the full transition relation of \texttt{inc\_xy}, which 
  includes approximating the outputs for the inputs, with which the procedure is 
  never called.
  If we analyse this program in a polyhedral domain, we will notice that
  the full input-output relation of \texttt{inc\_xy} cannot be approximated in a 
  useful way by a single convex polyhedron.
  But if we restrict the analysis to the reachable states, where $x \geq 0$ 
  always holds, we will be able to infer that \texttt{inc\_xy} increments both 
  $x$ and $y$, and this will allow to prove safety of the program.
\end{exmpl}

One may argue that we should alter the way we encode procedures and constrain 
$f_c$ to denote the set of reachable states at the entry of \texttt{inc\_xy}.
But when building an analysis tool, we should cater for different possible 
encodings.

\subsection{Combination of Forward and Backward Program Analyses}
\label{sec:background/qa}

Example~\ref{ex:parallel-increment-interproc} demonstrates the general problem 
of communicating analysis results between different program 
locations.
In an inter-procedural analysis, often we do not want to explicitly build the 
full input-output relation of a procedure.
For the inputs, with which a procedure may be called, we \emph{do} want to find 
the corresponding outputs, but for the other inputs we may want to report that 
the output is unknown.
This is because often, as in Example~\ref{ex:parallel-increment-interproc}, the 
full input-output relation will not have a useful approximation as a domain 
element.
At the same time, a useful approximation may exist when we consider only 
reachable inputs.
Similar considerations hold for intra-procedural analysis.
If we want to prove that an assertion violation is unreachable, we do not need
to explicitly represent the full inductive invariant of a program.
Instead, we want to approximate the set of states that are both reachable from 
the initial states and may reach an assertion violation.
If this set turns out to be empty, we can conclude that an assertion violation 
is unreachable.
This technique is standard for program analysis, and in \autoref{sec:main}, we 
adapt it to CHCs.

An alternative technique for Horn clauses is query-answer transformation 
\cite{gallagher15qa}.
Given the original system of CHCs $\HSys$, we build the transformed system 
$\HSys^{\rm qa}$.
For every uninterpreted predicate $p$ in $\HSys$ (including $\pfalse$), 
$\HSys^{\rm qa}$ contains a query predicate $p^q$ and an answer predicate $p^a$.
The clauses of $\HSys^{\rm qa}$ are constructed as follows.
\begin{compactitem}
  \item \emph{Answer clauses}. For every clause $\hc{P_1, \cdots, 
  P_n}{\phi}{P_{n+1}}$ (where $n \geq 0$) in $\HSys$, the system $\HSys^{\rm 
  qa}$ contains the clause $\hc{P_{n+1}^q, P_1^a, \cdots, 
  P_n^a}{\phi}{P_{n+1}^a}$.
  \item \emph{Query clauses}. For every clause $\hc{P_1, \cdots, 
  P_n}{\phi}{P_{n+1}}$ (where $n \geq 0$) in $\HSys$, the system $\HSys^{\rm 
  qa}$ contains the clauses:
  \[
  \begin{split}
    & \hc{P_{n+1}^q}{\phi}{P_1^q} \\
    & \hc{P_{n+1}^q, P_1^a}{\phi}{P_2^q} \\
    & \cdots \\
    & \hc{P_{n+1}^q, P_1^a, \cdots, P_{n-1}^a}{\phi}{P_n^q}
  \end{split}
  \]
  \item \emph{Goal clause} $\hc{\true}{}{\pfalse^q}$.
\end{compactitem}
Then, forward (bottom-up) analysis of $\HSys^{\rm qa}$ corresponds to a 
combination of forward and backward (top-down) analyses of $\HSys$.

We experienced several issues with the query-answer transformation.
For linear systems of CHCs, forward analysis of $\HSys^{\rm qa}$ corresponds to 
a single run of backward analysis of $\HSys$ followed by a single run of forward 
analysis.
For non-linear systems, this gets more complicated, though, as there will be 
recursive dependencies between query and answer predicates, and the propagation 
of information will depend on the order, in which query clauses are created.
We observed that is not enough, and for some systems the analysis needs to 
propagate the information forward and then backward multiple times.
This usually happens when the abstract domain of the analysis cannot capture the 
relation between the program variables.

\newsavebox{\additionLoopsText}
\begin{lrbox}{\additionLoopsText}
  \begin{lstlisting}
x = 0; y = *;
while(*)
  x += y;
if (x > 0) {
  while(*)
    y += x;
  assert(y $\geq$ 0);
}
  \end{lstlisting}
\end{lrbox}
\begin{wrapfigure}{L}{.4\textwidth}
  \centering
  \usebox{\additionLoopsText}
  \caption{Program, where polyhedral analysis needs to propagate information 
  forward and backward multiple times.}
  \label{fig:addition-loops-text}
\end{wrapfigure}

\begin{exmpl}{ex:fwd-backw-alternation}
  In \autoref{fig:addition-loops-text}, we show a synthetic example of a program 
  that needs more than one alternation of forward and backward analysis to be 
  proven safe.
  Notice that this program is safe, as after entering the \texttt{if}-branch in 
  line 4 we have that $x > 0 $ and $x = ky$ for some $k \geq 0$, therefore $y$ 
  is also greater than 0, and this is not changed by adding $x$ to $y$ in lines 
  5-6.
  If we work in a polyhedral domain, we cannot capture the relation $\exists k 
  \geq 0.\,x = ky$ and therefore should proceed with the safety proof in a 
  different way, e.g., as follows.
  First, we run a forward analysis and establish that at lines 5-7, $x > 0$, 
  since these lines are inside the \texttt{if}-branch.
  Then, we run a backward analysis starting with the set of states $y < 0$ at 
  line 7, which corresponds to the assertion violation.
  Since the loop in lines 5-6 can only increase $y$, we establish that for $y$  
  to be less than zero in line 7, it also has to be less than zero in lines 1-6.
  Finally, we run forward analysis again and establish that for the assertion 
  violation to be reachable, $x$ at line 4 has to be both greater than zero (so 
  that we enter the \emph{if}-branch), and less-or-equal to zero (because $x$ 
  starts being zero and in lines 2-3 we repeatedly add a negative number to it), 
  which is not possible.
  While this particular example is synthetic, in our experiments we observe a 
  small number of SV-COMP programs where a similar situation arises.
\end{exmpl}

A more subtle (but more benign) issue is that when solving the 
query-answer-transformed system, we are actually not interested in the elements 
of the interpretation of $p^a$, which are outside of $p^q$, but this is not 
captured in $\HSys^{\rm qa}$ itself.
Because of this, $p^a$ may be over-approximated too much as a result of widening 
or join.
Perhaps this is one of the reasons why B. Kafle and J. P. Gallagher propose 
\cite{gallagher15qa} to perform abstract interpretation in two phases.
First, they analyse the transformed system $\HSys^{\rm qa}$.
Then, they strengthen the original system with the interpretations of 
\emph{answer} predicates and run an analysis on the strengthened system.

To address these issues, we decided to adapt the standard (for program analysis) 
alternation of forward and backward analysis to CHCs.
We return to the comparison of our approach to query-answer transformation in 
\autoref{sec:main/revisiting-qa}.

\section{Combining Forward and Backward analysis of CHCs}
\label{sec:main}

Patrick and Radhia Cousot proposed a backward (top-down) semantics for Horn 
clauses, which
collects atoms that can appear in an SLD-resolution proof 
\cite{cousot92logic-programs}.
We take their definition as a starting point and define a new backward semantics 
and a new more precise combined forward-backward semantics.
Then we show, how we can use our new semantics to disprove reachability of a 
goal and to refine a model w.r.t. the goal.

\bparagraph{Backward Transformers and Collecting Semantics.}
First, let us introduce the pre-condition operation as follows.
For a system $\HSys$,
\[
\pre{\TRel}{A'} = \{ a \mid \exists A \subseteq \Atoms.\,\exists a' \in A'. (A, 
a') \in \TRel \land a \in A \}
\]
Then, for a system $\HSys$ and a set of goal atoms $\goal$, the backward 
(top-down) semantics is characterized by the least fixed point:
\begin{equation}
  \label{eq:backward-semantics-simple}
\lfplo{\subseteq}{X}{\goal \union \pre{\TRel}{X}}
\end{equation}
which corresponds to the semantics proposed by Patrick and Radhia Cousot.
This definition of backward semantics has a drawback though.
The intersection of forward semantics \eqref{eq:forward-semantics} and backward 
semantics \eqref{eq:backward-semantics-simple} \emph{over-approximates} the set 
of atoms that can be derived from initial clauses (of the form $\hc{}{\phi}{P}$) 
and can be used to derive the goal.
\begin{exmpl}{ex:overapproximate-intersection}
  Let us consider the following system of CHCs, where $p$ is a unary predicate 
  and $c_1,\cdots,c_5$ are constants
  \begin{equation}
    \label{eq:ex-overapproximate-intersection}
    \begin{aligned}
      & \hc{}{\true}{p(c_1)} & \hc{p(c_3)}{}{p(c_5)} \\
      & \hc{p(c_1)}{}{p(c_2)} \qquad & \hc{p(c_2),p(c_4)}{}{p(c_5)} \\
      & \hc{p(c_1)}{}{p(c_3)} \\
    \end{aligned}
  \end{equation}
  The forward semantics \eqref{eq:forward-semantics} for this system is the set 
  $\{ p(c_1), p(c_2), p(c_3), p(c_5) \}$
  (note that the atom $p(c_4)$ cannot be derived).
  Let us assume that the set of goals is $\goal = \{p(c_5)\}$.
  Then, the backward semantics \eqref{eq:backward-semantics-simple} for this 
  system is $\{ p(c_1), p(c_2), p(c_3), p(c_4), p(c_5) \}$.
  The intersection of forward and backward semantics is $\{ p(c_1), p(c_2), 
  p(c_3), p(c_5) \}$, even though the atom $p(c_2)$ is not used when deriving 
  the goal $\{p(c_5)\}$ (because we cannot derive $p(c_4)$).
  If we implement an abstract analysis based on the intersection of semantics  
  \eqref{eq:forward-semantics} and \eqref{eq:backward-semantics-simple}, this 
  will become an additional source of imprecision.
\end{exmpl}

\subsection{Forward and Backward Analyses Combined}
\label{sec:main/forward-backward-combined}

We wish to define a combination of forward and backward semantics that does not 
introduce the over-approximation observed in 
Example~\ref{ex:overapproximate-intersection}.
For that, we propose the \emph{restricted pre-condition} operation that we 
define as follows.
For a restricting set $R \subseteq \Atoms$,
\[
\preR{R}{\TRel}{A'} = \{ a \mid \exists A \subseteq R.\,\exists a' \in A'.\ (A, 
a') \in \TRel \land a \in A  \}
\]
Now, we can define the combined forward-backward collecting semantics as 
follows:
\begin{equation}
  \label{eq:fwd-backw-sem}
  \begin{split}
    & \lfplo{\subseteq}{X}{(\goal \isect M) \union \preR{M}{\TRel}{X}} \\
    & \text{where } M = \lfplo{\subseteq}{X}{\post{\TRel}{X}}
  \end{split}
\end{equation}
One can show that this semantics denotes the set of atoms that can be derived 
from initial clauses (of the form $\hc{}{\phi}{P}$) and can be used to derive 
the goal (we defer an explanation until \autoref{sec:tree-semantics}).
For example, one can see that for the system 
\eqref{eq:ex-overapproximate-intersection} discussed in 
Example~\ref{ex:overapproximate-intersection}, computing this semantics produces 
the set $\{p(c_1), p(c_3), p(c_5)\}$, as expected.

Introducing a restricted pre-condition operation is common, when a combination 
of analyses cannot be captured by the meet operation in the domain.
For example, assume that we want to analyse the instruction $z := x+y$ in an 
interval domain.
Assume also that the \emph{pre}-condition is restricted by $x \geq 3$ (e.g.,  
obtained by forward analysis) and the \emph{post}-condition is $z \in [0,2]$.
In this case, \emph{unrestricted} backwards analysis yields no new results.
But if we modify the pre-condition operation to take account of the previously 
obtained pre-condition ($x \geq 3$ in this case), we can derive the new 
constraint $y \leq -1$.

It may however be unusual to see a restricted pre-condition in concrete 
collecting semantics.
To explain it, in \autoref{sec:tree-semantics}, we introduce tree semantics of 
CHCs and show how concrete collecting semantics is itself an abstraction of tree 
semantics.
In particular, the intersection of forward and backward tree semantics abstracts 
to \eqref{eq:fwd-backw-sem}.

\bparagraph{Abstract Transformers.}
As standard in abstract interpretation, we introduce over-approximate versions 
of forward and backward transformers, resp. $\DpostName$ and $\DpreName$, s.t.
for $d, r \in \D$,
\[
  \concrete{\Dpost{\HSys}{d}} \supseteq \post{\TRel}{\concrete{d}}
  \qquad \concrete{\DpreR{r}{\HSys}{d}} \sqsupseteq 
  \preR{\concrete{r}}{\TRel}{\concrete{d}}
\]

\bparagraph{Abstract Iteration Sequence.}
In concrete world, the combination of forward and backward analyses is 
characterized by a pair of fixed points in \eqref{eq:fwd-backw-sem}.
In particular, we have the following property:
\begin{proposition}
  \label{lm:concrete-2-steps}
If we let $M = \lfplo{\subseteq}{X}{\post{\TRel}{X}}$ and $M' = 
\lfplo{\subseteq}{X}{(\goal \isect M) \union \preR{M}{\TRel}{X}}$ then 
$\lfplo{\subseteq}{X}{(\post{\TRel}{X} \isect M')} = M'$.
\end{proposition}
That is, concrete forward and backward analyses need not be iterated.
We give the proof of this a bit later.
In the abstract world, this is not the case, as has already been noted for 
program analysis \cite{cousot99refining-mc}.
In general, given the abstract goal $g \in \D$, the combination of abstract 
forward and backward analyses produces the sequence:
\begin{equation}
  \label{eq:fwd-backw-sequence}
\begin{split}
  & \db_0,\,\df_1,\,\db_1,\,\df_2,\,\db_2,\,\cdots \text{ , where} \\
  & \db_0 = \Dtop, \text{ and for } i\geq 1, \\
  & \Dpost{\HSys}{\df_i} \Dmeet \db_{i-1} \Dleq \df_i \\
  & \dgoal \Dmeet \df_i \Dleq \db_i \\
  & \DpreR{\df_i}{\HSys}{\db_i} \Dleq \db_i \\
\end{split}
\end{equation}
In principle, this iterations sequence may be infinitely descending, and to 
ensure termination of an analysis, we have to limit how many elements of the 
sequence are computed.
In our experiments though, the sequence usually stabilizes after the first few 
elements.

Propositions \ref{lm:fwd-backw-model} and \ref{lm:fwd-backw-disjoint} 
respectively  show how we can refine the initial model w.r.t. the goal and how 
we can use the iteration sequence to disprove reachability of the goal.

\begin{restatable}{proposition}{lmFwdBackwModel}
  \label{lm:fwd-backw-model}
  For every $k \geq 1$, the set
  $
  \concrete{\df_k} \union \bigunion_{i=1}^{k-1} \big(\concrete{\df_i} \setminus 
  \concrete{\db_i}\big)
  $
  is a model of $\HSys$.
\end{restatable}
We present the proof in Appendix~\ref{apx:proofs}.

Observe that for some abstract domains (e.g., common numeric domains: intervals, 
octagons, polyhedra), the meet operation is usually exact, i.e.  for $d_1, d_2 \in 
\D$, $\concrete{d_1 \Dmeet
d_2} = \concrete{d_1} \isect \concrete{d_2}$.
Also, for such domains we can expect that for $r,d \in D$, $\DpreR{r}{\HSys}{d} 
\sqsubseteq r$.
In this case, the forward-backward iteration sequence is descending: $\db_0 
\Dgeq \df_1 \Dgeq \db_1 \Dgeq \df_2 \Dgeq \cdots$, and computing every 
subsequent element $\df_i$ provides a tighter model of $\HSys$ (assuming $\df_i$ 
is distinct from $\df_{i-1}$).
This comes at a cost, though, since the refined model will not in general be 
expressible in the abstract domain of the analysis.
For example, in a polyhedral analysis, when $\df_i$ and $\db_i$ are maps from 
predicates to convex polyhedra, expressing the model given by
\autoref{lm:fwd-backw-model}, requires finite sets of convex polyhedra.
If we wish to check if such an object $M$ is indeed a model of $\HSys$, we will 
need to check that $M$ \emph{geometrically covers} its post-condition.
This can be done using a polyhedra library that supports powerset domains and 
geometric coverage (e.g., Parma Polyhedra Library \cite{bagnara08ppl}) or with 
an SMT-solver .

Now, the proof of \autoref{lm:concrete-2-steps} becomes straightforward.
\begin{proof}[of \autoref{lm:concrete-2-steps}]
  Let $M'' = \lfplo{\subseteq}{X}{(\post{\TRel}{X} \isect M')}$, i.e. $M'' 
  \subseteq M'$ by definition.
  From \autoref{lm:fwd-backw-model}, $(M \setminus M') \union M'' \subseteq M$ 
  is a model of $\HSys$.
  Since $M$ is the smallest model, $(M \setminus M') \union M'' = M$ and $M'' = 
  M'$.
\end{proof}

\begin{proposition}
  \label{lm:fwd-backw-disjoint}
  If there exists $k\geq 1$, s.t. $\df_k = \Dbot$, then there exists a model $M$ 
  of $\HSys$, s.t. $M \isect \concrete{\dgoal} = \emptyset$ (i.e., the goal is 
  unreachable).
\end{proposition}
\begin{proof}
  If $\df_k = \Dbot$ then $\concrete{\df_k} = \emptyset$, and from 
  \autoref{lm:fwd-backw-model}, $M = \bigunion_{i=1}^{k-1} 
  \big(\concrete{\df_i} \setminus \concrete{\db_i}\big)$ is a model of $\HSys$.
  From \eqref{eq:fwd-backw-sequence}, it follows that for every $i$, 
  $\concrete{\dgoal} \isect \concrete{\df_i} \subseteq \concrete{\db_i}$, that 
  is $(\concrete{\df_i} \setminus \concrete{\db_i}) \isect \concrete{\dgoal} = 
  \emptyset$.
  This means that $M \isect \concrete{\dgoal} = \emptyset$.
\end{proof}
Thus, when there exists $k$ s.t. $\df_k = \Dbot$, we obtain a 
\emph{constructive} proof of unreachability of the goal that can later be 
checked.

\bparagraph{Result of the Analysis.}
Propositions \ref{lm:fwd-backw-model} and \ref{lm:fwd-backw-disjoint} provide a 
way to give additional information to the user of the analysis, apart from the 
verdict (\emph{safe} or \emph{potentially unsafe}).
Suppose, we compute the iteration sequence \eqref{eq:fwd-backw-sequence} up to 
the element $\df_k$ and then stop (whether because $\df_k = \Dbot$, or the 
sequence stabilized, or we reached a timeout, etc).
The object $\df_k$ in itself may not be interesting: it is not a model of 
$\HSys$, it is not a proof or a refutation of reachability of the goal.
If the user wishes to check the results of the analysis,
we may give them the whole iteration sequence up to $\df_k$.
Then, the user will need to confirm that the sequence indeed satisfies the 
conditions of \eqref{eq:fwd-backw-sequence}.
Alternatively, we may give the user the refined model of $\HSys$, i.e. some 
representation of $M = \concrete{\df_k} \union \bigunion_{i=1}^{k-1} 
\big(\concrete{\df_i} \setminus \concrete{\db_i}\big)$.
This will allow the user to not only check the model, but also, e.g., produce 
program invariants that can be used by another verification tool (e.g., Frama-C 
\cite{url-frama-c}, KeY \cite{key16book}, etc).
Representation of $M$ may require an abstract domain that is more expressive 
than the domain of the analysis, but may be more compact than the whole 
iteration sequence.
Alternatively, if $\df_i$ and $\db_i$ can be represented in logical form in some 
theory, so can $M$.

\bparagraph{Which Analysis Runs First.}
In the iteration sequence \eqref{eq:fwd-backw-sequence}, forward and backward 
analyses alternate, but which analysis runs first is actually not fixed.
We may start with forward analysis and compute $\df_1$ as normal, or we may take 
$\df_1 = \Dtop$ and start the computation with backward analysis.
A notable option is to do the first run of backward analysis in a more coarse 
abstract domain and switch to a more precise domain in subsequent runs.
For example, the initial run of backward analysis may only identify the 
predicates that can potentially be used to derive the goal:
\begin{equation}
  \label{eq:coarse-backw}
\begin{split}
  & \lfplo{\subseteq}{X}{\Pgoal \union \pre{\TPred}{X}}, \text{ where} \\
  & \Pgoal = \{ p \mid p(\v{c}) \in \goal \} \\
  & \TPred = \big\{ (\Pi, p') \mid \exists (A, a') \in \TRel.\,\Pi=\{p \mid 
  p(\v{c}) \in A\} \land a'=p'(\v{c'}) \big\}
\end{split}
\end{equation}
Then, we can take $\df_1 = \Dtop$, $\db_1$ to be some abstraction of 
\eqref{eq:coarse-backw}, and starting from $\df_2$, run the analysis with a more 
precise domain.
In program analysis, restricting attention to program locations that have a path 
to (i.e., are backward-reachable from) some goal location, is a known technique.
For example, K. Apinis, H. Seidl, and V. Voidani describe a sophisticated 
version of it \cite{seidl12swiss}.

\subsection{Revisiting the Query-Answer Transformation}
\label{sec:main/revisiting-qa}

In principle, the iteration sequence \eqref{eq:fwd-backw-sequence} can be 
emulated by an iterated simple query-answer transformation.
Let $\HSys$ be the original system of CHCs.
Let the element $\db_k$ of the iteration sequence \eqref{eq:fwd-backw-sequence} 
map every predicate $p_i$ to a formula $\beta_k^i$.
In particular, $\db_0$ will map every $p_i$ to $\true$.
Then, $\df_{k+1}$ can be found as a model of the system $\HSys^{\df}_{k+1}$.
To construct, $\HSys^{\df}_{k+1}$, for every CHC
$\hc{P_1,\cdots,P_n}{\phi}{P_{n+1}}$ (for $n \geq 0$) in the original system 
$\HSys$, we add to $\HSys^{\df}_{k+1}$ the clause $\hc{P_1,\cdots,P_n}{\phi 
\land \beta_k^{n+1}}{P_{n+1}}$.
Now let the element $\df_k$ map every $P_i$ to a formula $\delta_k^i$.
Then, $\db_k$ can be found as a model of the system $\HSys^{\db}_k$ that is 
constructed as follows.
For every CHC in the original system $\HSys$: 
$\hc{P_1,\cdots,P_n}{\phi}{P_{n+1}}$, we add to $\HSys^{\db}_{k}$ the clauses
$\hc{P_{n+1}}{\phi \land \bigand_{i=1}^n \delta_k^i}{P_1}$ through 
$\hc{P_{n+1}}{\phi \land \bigand_{i=1}^n \delta_k^i}{P_n}$.
Also, we add to $\HSys^{\db}_{k}$ the goal clause $\hc{\pfalse_k}{}{\pfalse}$.
If we compute the elements of the iteration sequence up to $\df_k$, then the 
function that maps every $p_i$ to $\delta_k^i \lor \bigor_{j=1}^{k-1} ( 
\delta_j^i \land \neg\beta_j^i )$ represents a model of the original system 
$\HSys$.
In particular, when $\df_1 = \Dtop$, and $k=2$, this produces a model, where 
every  $p_i$ maps to $\beta_1^i \limpl \delta_2^i$.

Thus, one has a choice, whether to take a fixpoint-based approach, as we did, or 
a transformation-based approach.
From the theoretical point of view, one will still have to prove that the 
iterated transformation allows to prove unreachability of the goal and to build 
a refined model, i.e., some analog of Propositions \ref{lm:fwd-backw-model} and 
\ref{lm:fwd-backw-disjoint}.
As one can see in Appendix~\ref{apx:proofs}, this is not trivial for the steps 
beyond the second.
From the practical point of view, we believe that our approach allows to more 
easily implement some useful minor features.
For example, the iteration sequence \eqref{eq:fwd-backw-sequence} naturally 
constrains $\db_i$ to be below $\df_i$ and $\df_i$ to be below $\db_{i-1}$, 
which in some cases makes widening and join less aggressive.
It should be possible though to achieve a similar effect for the query-answer 
transformation at the expense of introducing additional predicates and clauses.

On the other hand, an advantage of query-answer transformation is that it can be 
used as a preprocessing step for the analyses that are not based on abstract 
interpretation.
For example, B. Kafle and J. P. Gallagher report \cite{gallagher15qa} that it 
can improve the precision of a CEGAR-based analyser.

\section{Implementation and Experiments}
\label{sec:experiments}

We implemented our approach in a prototype abstract interpreter.
It can analyse numeric C programs that were converted to a system of CHCs with 
the tool SeaHorn \cite{gurfinkel15seahorn} (the input format is currently a 
technical limitation, and we wish to remove it in the future).
The implementation is written in OCaml and available online \cite{url-hcai}.
A notable feature of SeaHorn is that it introduces Boolean variables and 
predicate arguments even for programs without Boolean variables.
To represent sets of valuations of numeric and Boolean variables, we use 
Bddapron \cite{url-bddapron}.
We implement Kleene-like iteration as outlined in 
\autoref{sec:background/chc-ai}, which is similar to path focusing 
\cite{david11path-focusing,david12pagai}.
Iteration order and choice of widening points are based on F.~Bourdoncle's 
\cite{bourdoncle93iteration,Bourdoncle_PhD} recursive strategy (except that we implement it 
using a worklist with priorities).
As an SMT solver, we use Z3 \cite{url-z3}.
For comparison, in addition to the forward-backward iteration sequence 
\eqref{eq:fwd-backw-sequence}, we implemented an analysis based on query-answer 
transformation.

To evaluate our implementation, we took C programs from the categories 
\emph{loops} and \emph{recursive} of the Competition on Software Verification 
SV-COMP \cite{url-sv-comp}.
SeaHorn operates on LLVM bytecode produced by Clang 
\cite{url-clang}, and the resulting system of CHCs depends a lot on Clang 
optimization settings.
For example, constant folding may remove whole computation paths when they do 
not depend on non-deterministic inputs.
Or, Clang may replace recursion with a loop, which will make SeaHorn produce a 
linear system of CHCs instead of a non-linear one.
In our experiments, we compiled the input programs with two optimization levels: 
\texttt{-O3} (SeaHorn's default) and \texttt{-O0}.
As a result, we get a total of 310 systems of Horn clauses, out of which 158 are 
declared safe by SV-COMP.
Since we cannot prove unsafety, our evaluation focuses on safe systems.
Out of 158 safe systems, our tool can work with 123.
Other systems use features that are not yet supported in our tool (division,
non-numeric theories, etc).
Out of 158 safe systems, 74 are non-linear.

First, we evaluate the effect of combined forward-backward analysis.
The results are presented in \autoref{tab:svcomp-abstints}.
We compare three approaches.
The first is the one we propose in this paper, i.e., based on the 
forward-backward iteration sequence \eqref{eq:fwd-backw-sequence}.
We compute the elements of \eqref{eq:fwd-backw-sequence} up to $\df_5$.
If we decrease the limit from $\df_5$ to $\df_3$, we can prove safety of 2 less 
programs; increasing the limit to $\df_7$ gives no effect.
The second one a 2-step analysis based on query-answer transformation 
\cite{gallagher15qa}.
First, it runs forward analysis on a query-answer transformed system, then 
injects the interpretations of answer predicates in the original system and runs 
forward analysis again.
We implemented this analysis ourselves, and thus we are \emph{not} 
directly comparing our implementation to the tool Rahft \cite{gallagher16rahft}, 
where this analysis was first implemented.
Finally, we also run a simple forward analysis.
In \autoref{tab:svcomp-abstints}, we report the number of programs that we 
\emph{proved safe} with every approach.
One can see that our approach has a small advantage over both query-answer 
transformation and simple forward analysis.
Interestingly, B. Kafle and J. P. Gallagher report \cite{gallagher15qa} a 
\emph{much}  greater difference when moving from simple forward analysis to 
query-answer transformation.
This can be attributed to three factors.
First, their set of benchmarks is different, although it includes many programs 
from the same SV-COMP categories.
Second, their benchmarks are, to our knowledge, not pre-processed by Clang.
Third, as B. Kafle and J. P. Gallagher themselves report, some issues solved by 
adding backward analysis can as well be solved by path focusing, which our tool 
implements.

\begin{table}[b]
  \centering
  \begin{minipage}[b]{.45\textwidth}
    \centering
    \begin{tabular}{c|c|c|c|c}
      & & \multicolumn{3}{c}{Proven safe} \\
      Safe & Supported & This paper & QA & Fwd. \\\hline
      158 & 123 & 87 & 82 & 76\\
    \end{tabular}
    \caption{Comparison of abstract interpretation strategies.}
    \label{tab:svcomp-abstints}
  \end{minipage}
  \hfill
  \begin{minipage}[b]{.475\textwidth}
    \centering
    \begin{tabular}{c|c|c}
      & This paper & SeaHorn \\\hline
      Proven safe & 87 / 123 (70\%) & 133 / 158 (84\%) \\
    \end{tabular}
    \caption{Comparison to SeaHorn's builtin solver (with 1 minute timeout).}
    \label{tab:svcomp-seahorn}
  \end{minipage}
\end{table}

For reference, we also compare our tool to the solver that is integrated with 
SeaHorn (to our knowledge, it is based on the tool SPACER.
\cite{komuravelli13spacer,komuravelli14spacer}).
We present the results in \autoref{tab:svcomp-seahorn}.
SeaHorn can prove safety of more programs, which is expected since our tool is 
in an early stage of development.

\section{Tree Semantics of CHCs}
\label{sec:tree-semantics}

In this section, we briefly introduce tree semantics of CHCs.
Trees are not convenient objects to work with, and studying tree semantics is 
not the main purpose of this paper.
Thus, our description will not be fully rigorous.
Rather, our goal is to give the reader an intuition of why we construct 
collecting semantics (especially, backward and combined semantics) in the way we 
do, which is perhaps best explained when collecting semantics is viewed as an 
abstraction of tree semantics.

For the purpose of this section, a \emph{tree} is either a leaf node containing 
an atom, or an interior node that contains an atom and also has a non-zero 
number of child subtrees.
\[
  \text{Tree } \Coloneqq \leaf{a} \mid \tree{a}{t_1, \cdots, t_n}
\]
where $a \in \Atoms$ and every $t_i$ is a tree.
The \emph{root atom} of a tree is naturally defined as
\[
  \root{\leaf{a}} = a \qquad
  \root{\tree{a}{t_1, \cdots, t_n}} = a
\]
The set of \emph{leaves} of a tree is defined as
\[
  \leaves{\leaf{a}} = \{ a \} \qquad
  \leaves{\tree{a}{t_1, \cdots, t_n}} = \bigunion_{i=1}^n \leaves{t_i}
\]
The tree semantics of a system of CHCs $\HSys$ is a set of trees, where the 
parent-child relation is defined by the direct consequence relation $\TRel$.
To get more formal, let us first define the post-condition operation on trees as 
follows:
\[
\Tpost{\HSys}{X} = \begin{aligned}[t]
    & \big\{ \tree{a'}{t_1,\cdots,t_n} \mid t_1,\cdots,t_n \in X\\
    & \land \exists 
    (A,a') \in  \TRel.\,|A|=n\land A=\{\root{t_1},\cdots,\root{t_n} \} \big\} 
    \union {}\\
    & \{ \leaf{a} \mid (\emptyset,a) \in \TRel
  \}\end{aligned}
\]
Intuitively, the operation performs two distinct actions: (i) it produces a 
trivial tree $\leaf{a}$ for every initial transition $(\emptyset, a)$; and (ii) 
for every non-initial transition $(A, a')$, it creates every possible tree 
$\tree{a'}{t_1,\cdots,t_n}$, where $t_i$ are elements of $X$, and their roots 
correspond to distinct elements of $A$.
Then, we can define the \emph{forward tree semantics} of $\HSys$ as the least fixed 
point:
\[
\lfplo{\subseteq}{X}{\Tpost{\HSys}{X}}
\]
Intuitively, this is the set of trees, where leaves are initial atoms, and 
parent-child relation is defined by the direct consequence relation.
One can say that this is the set of derivation trees induced $\HSys$.
A notable property of forward tree semantics is that it is 
\emph{subtree-closed}, i.e., with every tree, it also contains all of its 
subtrees.

Let us now define the \emph{set-of-atoms} abstraction of a set of trees.
First, let us define an auxiliary predicate that tells whether an atom is a node 
of a tree.
\[
\begin{split}
  & \isnode{a}{\leaf{a'}} = (a = a') \\
  & \isnode{a}{\tree{a'}{t_1,\cdots,t_n}} = (a = a') \lor \bigor_{i=1}^{n} 
  \isnode{a}{t_i}
\end{split}
\]
Then, for a set of trees $T$, its set-of-atoms abstraction is
\[
\Tabstr{T} = \{ a \mid \exists t \in T.\,\isnode{a}{t} \}
\]
In particular, when $T$ is subtree-closed, one can show that
\begin{equation}
  \label{eq:subtree-closed-abstraction}
  \Tabstr{T} = \{ \root{t} \mid t \in T \}
\end{equation}
Let us observe that the set-of-atoms abstraction of the forward tree semantics 
is exactly the forward collecting semantics:
\begin{proposition}
  \label{lm:exact-fwd-abstraction}
  $
    \Tabstr{\lfplo{\subseteq}{X}{\Tpost{\HSys}{X}}} =
    \lfplo{\subseteq}{X}{\post{\TRel}{X}}
  $
\end{proposition}
\begin{proof}[sketch]
  This is an instance of \emph{exact fixed point abstraction} \cite[theorem 
  7.1.0.4]{cousot79systematic-design}, and to prove the proposition, we need to 
  show that
  \begin{equation}
    \label{eq:exact-post-abstraction}
    \Tabstr{\Tpost{\HSys}{T}} = \post{\TRel}{\Tabstr{T}}
  \end{equation}
  This is not true for an arbitrary $T$, but can be shown as true when $T$ is  
  subtree-closed, as it follows from \eqref{eq:subtree-closed-abstraction}.
  The $\TpostName$ operation preserves subtree-closure, thus 
  \autoref{lm:exact-fwd-abstraction} can be seen as a fixed point in the lattice 
  of subtree-closed sets, where \eqref{eq:exact-post-abstraction} holds and thus 
  exact fixed point abstraction holds as well.
\end{proof}

Let us now define the \emph{backward tree semantics}.
For a set of trees $T$, let $\Tpre{\HSys}{T}$ be the set of trees that are 
produced from trees in $T$ by replacing a single leaf containing $a' \in \Atoms$ 
with a subtree $\tree{a'}{a_1,\cdots,a_n}$, s.t. $a_1,\cdots,a_n$ are distinct, 
and $(a', \{a_1,\cdots,a_n\}) \in \TRel$.
Also let $\tgoal = \{ \leaf{a} \mid a \in \goal \}$.
Then, the backward tree semantics of $\HSys$ is the least fixed point
\[
\lfplo{\subseteq}{X}{\tgoal \union \Tpre{\HSys}{X}}
\]
Intuitively, this is the set of trees where the root is in $\goal$, and 
parent-child relation is defined by the direct consequence relation.

Let us define a \emph{pre-tree} of a tree $t$ to be an object that is a tree and 
that is produced by selecting a number (possibly, zero) of non-root interior 
nodes and replacing every such interior node $\tree{a}{t_1, \cdots, t_n}$ with 
the leaf $\leaf{a}$.
A notable property of backward tree semantics is that it is 
\emph{pre-tree-closed}, i.e., with every tree, it also contains all of its 
pre-trees.
One can show that when $T$ is pre-tree closed,
\[
  \Tabstr{T} = \bigunion\{ \leaves{t} \mid t \in T \}
\]
Similarly to the forward case, the set-of-atoms abstraction of the backward tree 
semantics is exactly the backward collecting semantics.
\begin{proposition}
  $
  \Tabstr{\lfplo{\subseteq}{X}{\tgoal \union \Tpre{\HSys}{X}}} =
  \lfplo{\subseteq}{X}{\goal \union \pre{\TRel}{X}}
  $
\end{proposition}
\begin{proof}[sketch]
  The proof idea is similar to that of \autoref{lm:exact-fwd-abstraction}.
  We need to show that
  $
  \Tabstr{\tgoal \union \Tpre{\HSys}{T}} = \goal \union \pre{\TRel}{\Tabstr{T}}
  $
  which does hold when $T$ is pre-tree-closed; and pre-tree-closure is preserved 
  by the transformer $\lfun{X}{\tgoal \union \Tpre{\HSys}{X}}$.
\end{proof}

Now, let us consider the intersection of the forward and backward tree 
semantics: $\big(\lfplo{\subseteq}{X}{\Tpost{\HSys}{X}}\big) \isect 
\big(\lfplo{\subseteq}{X}{\tgoal \union \Tpre{\HSys}{X}}\big)$.
This is the set of trees that have initial atoms as leaves and a goal atom as 
root.
We can now observe that the combined forward-backward semantics 
\eqref{eq:fwd-backw-sem} is exactly the set-of-atoms abstraction of this object.

\begin{proposition}
  $
  \begin{aligned}[t]
    & \TabstrName\Big(\big(\lfplo{\subseteq}{X}{\Tpost{\HSys}{X}}\big) \isect 
    \big(\lfplo{\subseteq}{X}{\tgoal \union \Tpre{\HSys}{X}}\big)\Big) \\
    = {} & \lfplo{\subseteq}{X}{(\goal \isect M) \union \preR{M}{\TRel}{X}} \\
    & \text{where }  M = \lfplo{\subseteq}{X}{\post{\TRel}{X}}
  \end{aligned}
  $
\end{proposition}
To see intuitively why this is true, let $t \in 
\big(\lfplo{\subseteq}{X}{\Tpost{\HSys}{X}}\big) \isect 
\big(\lfplo{\subseteq}{X}{\tgoal \union \Tpre{\HSys}{X}}\big)$ and let us 
observe which atoms may appear in $t$ at different depth.
We know that $\root{t} \in \goal \isect M$.
At depth one, we will observe sub-trees that have initial atoms as leaves and 
can be combined to produce $t$.
One can see that the set of atoms at depth one is $\preR{M}{\TRel}{\goal \isect 
M}$.
Similarly, the set of atoms at depth two is 
$\preR{M}{\TRel}{\preR{M}{\TRel}{\goal \isect M}}$.
Continuing this way, we get that the set-of-atoms abstraction of the 
intersection of forward and backward tree semantics is 
$\lfplo{\subseteq}{X}{(\goal \isect M) \union \preR{M}{\TRel}{X}}$.

To summarise, the combined forward-backward semantics \eqref{eq:fwd-backw-sem} 
is the set-of-atoms abstraction of the intersection of forward and backward tree 
semantics.
Since set-of-trees intersection and set-of-states abstraction do not commute, 
we need to introduce the restricted pre-condition operation to define 
the combined semantics.

\section{Related Work}

Combining forward and backward analyses is standard when analysing programs.
A good explanation of the technique is given by Patrick and Radhia Cousot 
\cite[section 4]{cousot99refining-mc}.
They also propose to use it for the analysis of logic programs 
\cite{cousot92logic-programs}.
Their combination is an intersection of forward and backward collecting 
semantics.

F. Benoy and A. King were perhaps the first to apply abstract interpretation in 
a \emph{polyhedral domain} to constraint logic programs 
\cite{king96polyhedra-clp}.
J.  P. Gallagher et al. in a series of works (see, e.g., 
\cite{gallagher02polyhedra-clp,gallagher15qa}) apply it to specialized CLPs or 
CHCs.
Previous sections discuss the differences between their approach and ours.
Later work by B. Kafle, J. P. Gallagher, and J. F. Morales 
\cite{gallagher15tree-refinement,gallagher16rahft} introduces another analysis 
engine that is not based on abstract interpretation.
M. Proietti, F. Fioravanti et al. propose a similar analysis  
\cite{proietti14specialization} that iteratively specializes the initial system 
of CHCs by propagating constraints both forward and backward and by 
heuristically applying join and widening operators.
This process is repeated until the analysis arrives at a system that can be 
trivially proven safe or a timeout is reached.
Notably, this analysis avoids explicitly constructing the model of the original 
system.

Multiple researchers were advocating using Horn clauses for program 
verification,
Including A. Rybalchenko \cite{rybalchenko12synthesizing-verifiers}, N.  
Bj{\o}rner, and others.
A survey was recently made by N. Bj{\o}rner, A.  Gurfinkel, K.  McMillan, and A.  
Rybalchenko \cite{bgmr15horn-verification}.
Tools that allow to solve problems stated as systems of Horn clauses include 
E-HSF \cite{rybalchenko13existentially-quantified}, Eldarica 
\cite{rummer13eldarica}, Z3 (with PDR \cite{bjorner12z3-pdr} and SPACER 
\cite{komuravelli13spacer,komuravelli14spacer} engines), and others.
As our implementation is in early development, we do not make a detailed 
comparison to these tools.

Path focusing was described by D. Monniaux and L.  Gonnord 
\cite{david11path-focusing} and implemented by J. Henry, D. Monniaux, and M.  
Moy in a tool PAGAI \cite{david12pagai}.
This is an approach to abstract interpretation, where one uses an SMT solver to 
find a path through a program, along which to propagate the post-conditions.

\section{Conclusion and Future Work}

In this paper, we introduce a new backward collecting semantics, which is 
suitable for alternating forward and backward abstract interpretation of Horn 
clauses.
We show how the alternation can be used to prove unreachability of the goal and 
how every subsequent run of an analysis yields a refined model of the system.
Experimentally, we observe that combining forward and backward analyses is 
important for analysing systems that encode questions about reachability in C 
programs.
In particular, the combination that follows our new semantics improves the 
precision of our own abstract interpreter, including when compared to a forward 
analysis of a query-answer-transformed system.

We see the following directions for future work.
\emph{First}, we wish to be able to infer models that are \emph{disjunctive} in 
a meaningful way.
Currently, as we use Bddapron, we produce models where a predicate maps to a 
disjunctive formula, but the disjunctions are defined by the Boolean arguments 
of the predicate, which are often unrelated to the interesting facts about 
numeric arguments.
We wish to explore how partitioning approaches designed for program analysis 
\cite{rival07trace-partitioning,henry-david12succint-something} can be applied 
to the analysis of Horn clauses.
\emph{Second}, we note that currently, for the combination of forward and 
backward analyses to work, we need to explicitly specify the goal 
(\texttt{query}, in terms of SeaHorn language).
It would be nice though, if we could use the benefits of the combined analysis 
(e.g., analysing the procedures only for reachable inputs) without having an 
explicit goal.
For that, we will need to be able to distinguish, which of the clauses of the 
form $\hc{}{\phi}{P}$ denote the program entry (the \texttt{main()} function in 
C terms), and which correspond to the procedures (recall Figures 
\ref{fig:parallel-increment-interproc-text} and 
\ref{fig:parallel-increment-interproc-chc}).
So far, the only solution we see is that this information needs to be 
communicated to our analyser as part of the input.
\emph{Finally}, we observe that so far we evaluate our approach using CHCs that 
result from reachability questions in relatively simple C programs.
These CHCs are also relatively simple and in particular contain at most two 
predicate applications in the bodies.
We wish to evaluate our approach using more complicated CHCs, e.g., that result 
from cell morphing abstraction \cite{david16cell-morphing}, but successfully 
analysing such systems requires to be able to produce disjunctive models.

\bibliographystyle{splncs03}
\bibliography{sas17}

\appendix

\section{Proofs}
\label{apx:proofs}

\lmFwdBackwModel*

\begin{proof}
  For convenience, let us replace the direct consequence relation $\TRel$ with 
  two objects: the set of initial atoms $\TInit = \{ a' \mid (\emptyset, a') \in 
  \TRel\}$ and the set of consecutions $\TStep = \{ (A, a') \in \TRel \mid A 
  \neq \emptyset \}$.
  Then, for every $R,X \subseteq \Atoms$, $\post{\TRel}{X} = \TInit \union 
  \post{\TStep}{X}$ and $\preR{R}{\TRel}{X} = \preR{R}{\TStep}{X}$.

  Now let us consider the first three elements of the descending sequence, 
  $d_1$, $b_1$, and $d_2$.
  For $d_1$ it holds that $\TInit \union \post{\TStep}{\concrete{\df_1}} 
  \subseteq \concrete{\df_1}$.
  That is, $\concrete{\df_1}$ is a model of $\HSys$ and the lemma statement 
  holds for $k=1$.

  For $b_1$, it holds that $(\concrete{\dgoal} \isect \concrete{\df_1}) \union 
  \preR{\concrete{\df_1}}{\TStep}{\concrete{\db_1}} \subseteq \concrete{\db_1}$.
  This means that for every conseqution $(A, a') \in \TStep$, if $A \subseteq 
  \concrete{\df_1}$ and $A \isect (\concrete{\df_1} \setminus \concrete{\db_1}) 
  \neq \emptyset$, then $a' \in (\concrete{\df_1} \setminus \concrete{\db_1})$.

  Finally, for $\df_2$ it holds that $(\TInit \union 
  \post{\TStep}{\concrete{\df_2}}) \isect \concrete{\db_1} \subseteq \df_2$.
  First, this means that $\TInit \subseteq (\concrete{\df_1} \setminus 
  \concrete{\db_1}) \union \concrete{\df_2}$.
  Indeed, by definition of $\df_1$, $\TInit \subseteq \concrete{\df_1}$ and by 
  definition of $\df_2$, $\TInit \isect \concrete{\db_1} \subseteq 
  \concrete{\df_2}$.
  Second, this means that $\post{\TStep}{(\concrete{\df_1} \setminus 
  \concrete{\db_1}) \union \concrete{\df_2}} \subseteq (\concrete{\df_1} 
  \setminus \concrete{\db_1}) \union \concrete{\df_2}$.
  Indeed, let is pick an arbitrary $(A, a') \in \TStep$, s.t. $A \subseteq 
  (\concrete{\df_1} \setminus \concrete{\db_1}) \union \concrete{\df_2}$.
  There are two possible cases.
  If $A \subseteq \concrete{\df_2}$ then by definition of $\df_2$, either $a' 
  \in \concrete{\df_2}$, or $a' \in (\concrete{\df_1} \setminus 
  \concrete{\db_1})$.
  If $A \not\subseteq \concrete{\df_2}$ then $A \isect (\concrete{\df_1} 
  \setminus \concrete{\db_1}) \neq \emptyset$, and $a' \in \concrete{\df_1} 
  \setminus \concrete{\db_1}$.
  This proves the statement of the lemma for $k = 2$ and also provides the base 
  case for the following inductive proof.

  Now let $k > 2$, $L_k = \bigunion_{i=1}^{k-1} \big(\concrete{\df_i} \setminus 
  \concrete{\db_i}\big)$, and $M_k = \concrete{\df_k} \union L_k$.
  Let the induction hypothesis be that: $\TInit \subseteq M_k$, 
  $\post{\TStep}{M_k} \subseteq M_k$ (i.e., $M_k$ is a model of $\HSys$), and
  for every $(A, a') \in \TStep$, if $A \subseteq M_k$ and $A \isect L_k \neq 
  \emptyset$, then $a' \in L_k$.

  Then, let us consider the two subsequent elements: $\db_k$ and $\df_{k+1}$ and 
  the two sets: $L_{k+1} = M_k \setminus \concrete{\db_k}$ and $M_{k+1} = 
  L_{k+1} \union \concrete{\df_{k+1}}$.
  
  For $\db_k$ it holds that $(\concrete{\dgoal} \isect \concrete{\df_k}) \union 
  \preR{\concrete{\df_k}}{\TStep}{\concrete{\db_k}} \subseteq \concrete{\db_k}$.
  That is, for every $(A, a') \in \TStep$, if $A \subseteq \concrete{\df_k}$ and 
  $A \isect (\concrete{\df_k} \setminus \concrete{\db_k}) \neq \emptyset$, then 
  $a' \in (\concrete{\df_k} \setminus \concrete{\db_k})$.

  For $\df_{k+1}$ it holds that $(\TInit \union
  \post{\TStep}{\concrete{\df_{k+1}}}) \isect \concrete{\db_k} \subseteq 
  \concrete{\df_{k+1}}$.

  First, observe that $\TInit \subseteq M_{k+1}$.
  Indeed, we know that $\TInit \subseteq M_k$ and that $M_{k+1} = (M_k \setminus 
  \concrete{\db_k}) \union \concrete{\df_{k+1}}$.
  By definition of $\df_{k+1}$, $\TInit \isect \concrete{\db_k} \subseteq 
  \concrete{\df_{k+1}}$.
  Thus, $\TInit \subseteq M_{k+1}$.

  Second, let us pick an arbitrary $(A, a') \in \TStep$, s.t. $A \subseteq 
  M_{k+1}$.
  Since $M_k$ is a model of $H$, we know that $a' \in M_k$.
  But then, there are three possible cases.
  (i) If $A \subseteq \concrete{\df_{k+1}}$, then either $a' \in 
  \concrete{\df_{k+1}}$, or $a' \notin \concrete{\db_k}$.
  That is, $a' \in (M_k \setminus \concrete{\db_k}) \union \concrete{\df_{k+1}} 
  = M_{k+1}$.
  (ii) If $A \subseteq \concrete{\df_k}$ and $A \not\subseteq 
  \concrete{\df_{k+1}}$, then $A \isect (\concrete{\df_k} \setminus 
  \concrete{\db_k}) \neq \emptyset$, and $a' \in \concrete{\df_k} \setminus 
  \concrete{\db_k} \subseteq M_{k+1}$.
  (iii) Finally, if $A \not\subseteq \concrete{\df_k}$, then $A \isect L_k \neq 
  \emptyset$, and from the hypothesis $a' \in L_k$.
  There are no other possible cases.
  This means that $\post{\TStep}{M_{k+1}} \subseteq M_{k+1}$ and thus $M_{k+1}$ 
  is a model of $\HSys$.
  Also, from (ii) and (iii) it follows that for $(A, a') \in \TStep$, if $A 
  \subseteq M_{k+1}$ and $A \isect L_{k+1} \neq \emptyset$, then $a' \in 
  L_{k+1}$.
\end{proof}

\end{document}